\documentclass[a4paper,11pt]{article}

\usepackage{amssymb,amsmath,amsthm}
\usepackage{a4wide}
\usepackage{pgf,tikz}
\usetikzlibrary{arrows}
\usepackage{epsfig}
\usepackage{color}
\usepackage{appendix}
\usepackage{array}
\usepackage{authblk}
\usepackage{subcaption}
\usepackage{tikz}
\usepackage{tikz-cd}
\usepackage{caption}

\newtheorem{theorem}{Theorem}[section]

\newtheorem{proposition}[theorem]{Proposition}

\newtheorem{corollary}[theorem]{Corollary}

\newcommand{\beq}{\begin{equation}}
\newcommand{\eeq}{\end{equation}}

\newcommand{\ds}{\displaystyle}

\theoremstyle{definition}
\newtheorem{definition}[theorem]{Definition}

\newtheorem{remark}[theorem]{Remark}

\begin{document}


\title{On refactorization problems and rational Lax matrices of quadrirational Yang--Baxter maps }
\author[1]{Pavlos Kassotakis}
\author[2]{Theodoros E. Kouloukas}
\author[3]{Maciej Nieszporski}

\affil[1]{Department of Mathematical Methods in Physics, Faculty of Physics,
University of Warsaw, Pasteura 5, 02-093, Warsaw, Poland}

\affil[2] {School of Computing and Digital Media, London Metropolitan University}
\affil[3]{Department of Mathematical Methods in Physics, Faculty of Physics,
University of Warsaw, Pasteura 5, 02-093, Warsaw, Poland}
 
\maketitle

\begin{abstract}
 We present rational Lax representations for one-component parametric quadrirational Yang--Baxter maps in both the abelian and non-abelian settings. We show that from the Lax matrices of a general class of non-abelian involutive Yang–Baxter maps ($\mathcal{K}$-list), 
by considering the symmetries of the $\mathcal{K}$-list maps, we obtain compatible refactorization problems with rational Lax matrices for other classes of non-abelian involutive Yang--Baxter maps ($\Lambda$, $\mathcal{H}$ and $\mathcal{F}$ lists). In the abelian setting, this procedure generates rational Lax representations for the abelian Yang–Baxter maps of the $F$ and $H$ lists.
Additionally, we provide examples of non-involutive (abelian and non-abelian)  multi-parametric Yang--Baxter maps, along with their Lax representations, which lie outside the preceding lists.
\end{abstract}


\section{Introduction}

Yang--Baxter maps,  i.e. set theoretical solutions of the Yang--Baxter equation 
\cite{Baxter,buch,Drin,skly88,ves2,Yang}, play an important role in the theory of discrete integrable systems. The Yang-Baxter equation expresses a compatibility condition  associated with classical integrability features, such as Lax representations, conserved quantities, Bäcklund–Darboux transformations, invariant Poisson structures, symmetries and exact solutions (see e.g. \cite{ ABS2,AdYa,sokor,kp3,MPW,paptonves,ves2}). 

A type of quadrirational Yang--Baxter maps on $\mathbb{CP}^1 \times  \mathbb{CP}^1$ has been partitioned under the Yang--Baxter equivalence into two lists the $F$-list \cite{ABS2} and the $H$-list \cite{PSTV}.
These two lists do not provide a complete classification, as various examples of non-involutive Yang–Baxter maps of the same type exist that are not included there (see e.g. \cite{kp4}). This fact does not diminish the importance of $F$ and $H$ lists, since they include some of the most celebrated examples of Yang--Baxter maps  which are related to well-known integrable lattice equations and have very interesting geometric interpretation. As was shown in \cite{KasKoul}, the non-abelian counterparts of the $F$ and $H$ lists can be grouped into four distinct, non-equivalent lists: the $\mathcal{K},\Lambda, \mathcal{H}$ and $\mathcal{F}$ lists. In the abelian setting, all three lists $\mathcal{K}, \Lambda, \mathcal{H}$ become equivalent to the $H$-list, while $\mathcal{F}$ reduces to the $F$-list.

In this paper, we aim to present rational Lax representations of the abelian $F$ and $H$-list quadrirational Yang--Baxter maps and their non-abelian analogues ($\mathcal{K},\Lambda, \mathcal{H}$ and $\mathcal{F}$ lists).
Matrix refactorization problems and Lax matrices of Yang--Baxter maps first appear in \cite{ves4,ves2,ves3}, where a construction of Lax matrices by M{\"o}bius transformations on $\mathbb{CP}^1$ is also presented. 
However, a direct application of this construction does not always successfully produce rational Lax matrices. This is particularly obvious in the case of the $F$-list Yang--Baxter maps. To overcome this difficulty, we use the symmetries of these maps.

We begin with the more general case of the non-abelian Yang–Baxter maps in the $\mathcal{K}$-list. We present a map with additional free parameters that generates all the $\mathcal{K}$-list for specific constant values of these parameters \cite{KasKoul}. Hence, a rational Lax matrix associated with this map produces all the rational Lax matrices of the $\mathcal{K}$-list maps. The maps from the rest non-abelian cases are obtained through the symmetries of the $\mathcal{K}$-list. Based on these symmetries, we introduce refactorization problems involving rational Lax matrices for $\Lambda$, $\mathcal{H}$ and $\mathcal{F}$ lists. The abelian limit of the $\mathcal{K}$-list gives rise to rational Lax matrices for the $H$-list Yang-Baxter maps (except of $H_V$ which is treated separately), while the abelian limit of the $\mathcal{F}$-list produces compatible refactorization problems with rational Lax matrices for the $F$-list. This approach does not apply to the $F_{IV}$ map, as it lacks a counterpart in the $H$-list. 

Finally, we present non-involutive six parameter extensions of the non-abelian $\mathcal{F}_{III},\mathcal{F}_V,\mathcal{K}_{III},\mathcal{H}_V$ Yang--Baxter maps and their associated  Lax representations.

\subsection{Yang--Baxter maps and Lax matrices}
A map $R: \mathcal{X} \times
\mathcal{X} \rightarrow \mathcal{X} \times \mathcal{X}$, $R:(x,y)\mapsto (u(x,y),v(x,y))$, where $\mathcal{X}$ is a set, is called a {\it{Yang--Baxter (YB) map}}  if it satisfies the {\em
Yang-Baxter equation}
\begin{equation*}\label{YBprop}
 R_{23}\circ R_{13}\circ
R_{12}=R_{12}\circ R_{13}\circ R_{23} \;.
\end{equation*} Here,  $R_{ij}$
for $i,j=1,2,3$, denotes the action of the map $R$ on the $i$ and
$j$ factor of $\mathcal{X} \times \mathcal{X} \times \mathcal{X}$,
i.e. $R_{12}(x,y,z)=(u(x,y),v(x,y),z)$,
$R_{13}(x,y,z)=(u(x,z),y,v(x,z))$ and
$R_{23}(x,y,z)=(x,u(y,z),v(y,z))$.

A YB map $R:(\mathbb{X} \times
\mathbb{I}) \times (\mathbb{X} \times \mathbb{I}) \mapsto
(\mathbb{X} \times \mathbb{I}) \times (\mathbb{X} \times
\mathbb{I})$, with
\begin{equation} \label{pYB}
R:((x,p),(y,q))\mapsto((u,p),(v,q))= ((u(x,p,
y,q),p),(v(x,p, y,q),q)),
\end{equation}
is called a \emph{parametric YB map} (\cite{ves2,ves3}). In this definition the set $\mathcal{X}$ is the Cartesian
product of the set of variables $\mathbb{X}$ and the set of parameters $\mathbb{I}$
with elements $p,q \in \mathbb{I}$ which remain invariant under $R$.
Typically, we consider that the set of variables 
$\mathbb{X}$ and the set of parameters $\mathbb{I}$ have the structure of an algebraic variety. 
We mainly deal here with the case where,
$\mathbb{X}$ is the set of complex numbers $\mathbb{C}$ or more precisely Riemann sphere $\mathbb{CP}^1$
and the set of parameters  $\mathbb{I}$ is  $\mathbb{C}^n$ or $(\mathbb{CP}^1)^n$, where $n$ denotes the number of parameters.
However, we discuss also 
 non-abelian cases, where we assume that $\mathbb{X}$ is a division ring $\mathcal{A}$, while the parameters (including the spectral parameter that appears in the Lax matrices) will be considered as elements of the center of the ring.
 We often
keep the parameters separately and denote a parametric YB map as
$R^{p,q}(x,y):\mathbb{X}\times \mathbb{X} \rightarrow
\mathbb{X} \times \mathbb{X}$  treating it as family of maps.

According to \cite{ves4}, a {\em Lax Matrix} of the parametric YB
map (\ref{pYB}) is a map $L:\mathbb{X} \times \mathbb{I} \times \mathbb{I}\rightarrow Mat(n\times n)$, such
that
\begin{equation} \label{laxmat}
L(u,p,\zeta)L(v,q,\zeta)=L(y,q,\zeta)L(x,p,\zeta)\;,
\end{equation}
 holds for any $\zeta\in \mathbb{I}$. 
 If the converse is also true, i.e. if the condition that equation (\ref{laxmat}) holds for any 
$\zeta\in \mathbb{I}$ 
implies $(u,v)=R_{p,q}(x,y)$, then $L(x,p,\zeta)$ is called {\em
strong Lax matrix}.
The parameter $\zeta$, which does not appear in the map, is called the {\it spectral parameter}.

On the other hand, solutions of \eqref{laxmat} give rise to YB maps under the so-called {\it 3-factorization property}
of the matrix $L$, which states that if $u=u^{p,q}(x,y), v=v^{p,q}(x,y)$ satisfy (\ref{laxmat}), for a matrix $L$  and the equation $$L(
\hat{x},p,\zeta )L( \hat{y} ,q,\zeta )L(\hat{z}, \gamma ,\zeta)= L(x
,p,\zeta)L(y, q,\zeta)L(z, \gamma,\zeta)$$ implies that $\hat{x}=x, \
\hat{y}=y$ and $\hat{z}=z$, for every $x,y,z \in \mathbb{X}$, then
$R_{p,q}(x,y)\mapsto(u,v)$ is a parametric YB map with Lax
matrix $L$ \cite{kp1}.

As it was shown in \cite{ves4,ves3}, if a YB map is generated by an effective action of the linear group $GL_N$, then the YB equation
yields a Lax representation. 
Conversely, using similar arguments from \cite{ves4}, we can show that  the solutions to the refactorization problem \eqref{laxmat}, which are  expressed by actions of $GL_N$,  satisfy the 3-factorization property and thus are YB maps. 

\begin{proposition}[\cite{ves4}] \label{propMob}
We consider a parametric map $R^{p,q}: \mathbb{X} \times \mathbb{X} \rightarrow  \mathbb{X} \times \mathbb{X}$ with
$$R^{p,q}:(x,y)\mapsto (u^{p,q}(x,y),v^{p,q}(x,y)):=(u,v),$$  an effective{\footnote{ $L\in GL_N$ acts identically on $\mathbb{X}$, iff $L=Id$}} action $GL_N \times \mathbb{X}\rightarrow \mathbb{X}$
  and a matrix valued function $L: \mathbb{X}\times \mathbb{I} \times \mathbb{I} \rightarrow GL_N$ such that
\begin{equation} \label{uvMob}
u=L(y,q,p)[x],  \ v=L(x,p,q)[y],
\end{equation}
where $L[x]$ denotes the action of $L\in GL_N$ on $x\in \mathbb{X}$. Then the map
$R^{p,q}$ is a parametric YB map if and only if
\begin{equation} \label{priff}
L(x,p,\zeta)L(y,q,\zeta)=L(v,q,\zeta)L(u,p,\zeta)\;.
\end{equation}

\end{proposition}

\ 

In many cases, we can derive rational Lax matrices of YB maps on $\mathbb{CP}^1 \times  \mathbb{CP}^1$ by considering the group $GL_2$  acting on $\mathbb{CP}^1$ by M{\"o}bius transformations \cite{ves4,ves3}, i.e. for $L=\begin{pmatrix}
    a&b\\
    c&d
\end{pmatrix}
 \in GL_2$,  $L[x]:={\displaystyle\frac{a x+b}{cx+d}}\;$. 
For example, Adler's map ($H_V$ map in the list of \cite{PSTV}) is given by
\begin{equation} \label{adler}
u=y-\frac{p-q}{x+y}\;, \ v=x+\frac{p-q}{x+y}\;,
\end{equation}
which is derived by the M{\"o}bius transformations corresponding to
\begin{equation} \label{laxAdler}
L(x,p,\zeta):=\left(
\begin{array}{cc}
x& \ p - \zeta+x^2 \\
 1& x
\end{array}
\right).
\end{equation}
 Hence, $L(x,p,\zeta)$ is a Lax matrix of Adler's map \cite{ves4,ves2}.
However, this procedure of obtaining rational Lax pairs of rational YB maps is not always successful. For example,     the $F_{IV}$ YB map \cite{ABS2}
\begin{equation} \label{F4-YB}
u=y\left(1-\frac{p-q}{x-y}\right)\;, \ v=x\left(1-\frac{p-q}{x-y}\right)\;,
\end{equation}
is derived by the M{\"o}bius transformation that corresponds to the  Lax matrix 
\begin{equation} \label{laxf4}
L(x,p,\zeta):=\left(
\begin{array}{cc}
{\ds\sqrt{x}}&{\ds \sqrt{x}(p-x+\zeta) }\\ [3mm]
 {\ds \frac{1}{\sqrt{x}}}& -{\ds \sqrt{x}}
\end{array}\right),
\end{equation}
that is clearly not rational.

\subsection{ The H and the F list of quadrirational YB maps and their symmetries}

In order to distinguish between equivalence classes of YB maps we need an equivalence relation that is introduced in the following Proposition.

 \begin{proposition}[\cite{PSTV}] \label{equivalence}
 If $R^{p,q}:\mathbb{X}\times \mathbb{X} \rightarrow
\mathbb{X} \times \mathbb{X}$ is a parametric YB map and $\phi(p):\mathbb{X} \rightarrow
\mathbb{X}$ a family of bijections parametrised by $p \in \mathbb{I}$, then the map
\begin{equation} \label{equiv}
\tilde{R}^{p,q}= (\phi^{-1}({p})\times \phi^{-1}({q})) R^{p,q}  (\phi({p})\times \phi({q}))
\end{equation}
is a YB map. 
Equation (\ref{equiv})  establishes an equivalence relation in the set of YB maps and
if the YB maps $R$ and $\tilde{R}$ are related by  (\ref{equiv}) then we refer to them as {\it  equivalent}.
 \end{proposition}

The following definition introduces the notion of a symmetry of a YB map.

\begin{definition}[\cite{PSTV}]
     A symmetry of a YB map $R^{p,q}:\mathbb{X}\times \mathbb{X} \rightarrow
\mathbb{X}\times \mathbb{X}$ is a parametric  family of bijections   $\sigma({p}):\mathbb{X} \rightarrow
\mathbb{X}$, such that
\begin{equation} \label{sym}
(\sigma(p)\times \sigma(q))R^{p,q}=R^{p,q}( \sigma(p)\times \sigma(q)).
\end{equation}
\end{definition}
Note that if $\sigma(p)$ is a symmetry of a YB map $R_{p,q}$, then the maps
\begin{align*} 
R^{\sigma}=&(\sigma(p)^{-1}\times Id)R^{p,q}(Id \times \sigma(q)),& \hat R^{\sigma}=&(Id\times \sigma(q)^{-1})R^{p,q}(\sigma(p)\times Id),
\end{align*}
are parametric  YB maps as well. In general, neither $R^{\sigma}$ nor $R^{\sigma}$, are  YB equivalent with the map $R$. Note that when the YB map $R_{p,q}$ is involutive then $R^{\sigma}=\hat R^{\sigma}.$

\begin{definition}[\cite{ABS2}]
A map $R: \mathcal{X} \times\mathcal{X} \rightarrow \mathcal{X} \times \mathcal{X}$,
with $R:(x,y)\mapsto (u(x,y),v(x,y))$, is called {\it{quadrirational}} if both maps
$v(x,\cdot):\mathcal{X} \rightarrow \mathcal{X} $ and $u(\cdot,y):\mathcal{X} \rightarrow \mathcal{X}$, for fixed $x,y\in\mathcal{X}$ respectively, are
birational isomorphisms of $\mathcal{X}$.
\end{definition}

In \cite{ABS2},  Adler, Bobenko and Suris classified all parametric quadrirational maps of the subclass $[2:2]$  on
$\mathbb{CP}^1 \times  \mathbb{CP}^1$ into five cases of the so-called $F$-list under $(M\text{\"o}b)^4$ transformations, i.e.
M{\"o}bius transformations acting independently on each field $x,y,u,v$. All these cases turned out to be YB maps.
 However, in general the YB property is not preserved under $(M\text{\"o}b)^4$ transformations (quadrirational maps are not necessarily YB maps).  That allowed \cite{PSTV} to complement the $F-list$ with five additional cases that form the so-called $H$-list.
The maps in the $H$-list are derived from those in the $F$-list through the symmetries of the latter. Explicitly the $F$ and the $H-$list are presented in appendix \ref{app1}.

Applying Proposition \ref{propMob} to any member of the $F$ or the $H-$list, we obtain associated Lax pairs. Note though that  the obtained Lax pairs of the $H-$list are rational (an appropriate normalization might be of use) while the ones of the $F-$list (apart $F_{III}$ and $F_V$)
 are not.  

In Table \ref{table1}, we present the rational Lax matrices of the $H-$list, which are obtained by Proposition \ref{propMob}.  Indeed, it can be easily shown that the mappings $H_I-H_V$ (see Appendix \ref{app1}) can be expressed, via a $GL_2$ action, as (\ref{uvMob}) for the corresponding Lax matrices given in Table \ref{table1}. In addition, (\ref{priff}) for each $L$ given in  Table \ref{table1}, is equivalent to $H_I-H_V$. So  the matrices $L$ given in  Table \ref{table1}, serve as Lax matrices for the $H-$list of Yang-Baxter maps.

\begin{table}[htbp]
 \begin{tabular}{|c|c|c|c|c|c|}
  \hline
   & $H_I$&$H_{II}$&$H_{III}^A$&$H_{III}^B$&$H_V$ \\  \hline
  $L(x,p;\zeta)$ & ${\ds\begin{pmatrix}
                          \frac{x-p}{x-1}&\zeta (p-1)\frac{x}{x-1}\\
                          p-x& x
                          \end{pmatrix}}$ & ${\ds\begin{pmatrix}
                          1&\zeta p (x-1) \\
                          \frac{1}{x}& \frac{x-1}{x}\end{pmatrix}}$& ${\ds\begin{pmatrix}
                          1&\zeta p x \\
                          \frac{1}{x}& 1 \end{pmatrix}}$ & ${\ds\begin{pmatrix}
                         p x &\zeta \\
                         1 & \frac{1}{x}\end{pmatrix}}$ &  ${\ds\begin{pmatrix}
                          x&p-\zeta+x^2 \\
                          1 & x \end{pmatrix}}$    \\
  \hline
\end{tabular}
\caption{Lax matrices of the $H-$list of quadrirational YB maps} \label{table1}
\end{table}


\section{Lax matrices of YB maps which are related through symmetry}

The following  is the main Theorem of this article.

\begin{theorem} \label{propsym}
Let $R^{p,q}:(x,y)\mapsto (u,v),$ be a YB map with Lax matrix $L(x,p,\zeta)$ and $\sigma_p:=\sigma(p)$ a symmetry of this map.
Then the YB map, 
\begin{align}\label{Rs}
\hat R^{\sigma}:=(\sigma_{p}^{-1}\times Id)R^{p,q}(Id \times \sigma_{q}):(x,y) \mapsto (\hat u,\hat v),
\end{align}
satisfies the refactorization problem
\begin{equation} \label{refs0}
L(\sigma_{p}(\hat u),p,\zeta)L(\hat v,q,\zeta)=L(\sigma_{q}({y}),q,\zeta)L(x,p,\zeta)\;,
\end{equation}
while the YB map 
\begin{align}
\tilde R^{\sigma}:=(Id\times \sigma_{q}^{-1})R^{p,q}(\sigma_{p}\times Id):(x,y) \mapsto (\tilde u,\tilde v),
\end{align}
satisfies
\begin{equation}
   \label{refs1}
L(\tilde u,p,\zeta)L(\sigma_{q}(\tilde v),q,\zeta)=L(y,q,\zeta)L(\sigma_{p}(x),p,\zeta)\;. 
\end{equation}
Furthermore, if $L$ is a strong Lax matrix then equation \eqref{refs0} is equivalent to  $(\hat u,\hat v)=\hat R^{\sigma}(x,y)$, while \eqref{refs1} is equivalent to  $(\tilde u,\tilde v)=\tilde R^{\sigma}(x,y)$. 
\end{theorem}

\begin{proof}
Let $R^{p,q}(x,y):=(u(x,p,y,q),v(x,p,y,q))$, then from \eqref{Rs} we derive that 
\begin{equation} \label{prf1}
(\sigma_p(\hat u),\hat{v})=(u(x,p,\sigma_q(y),q),v(x,p,\sigma_q(y),q))\;.
\end{equation}
Now since $L(x,p,\zeta)$ is a Lax matrix of $R^{p,q}$, we have that 
\begin{equation} \label{prf2}
L(u(x,p,\sigma_q(y),q),p,\zeta)L(v(x,p,\sigma_q(y),q),q,\zeta)=
L(\sigma_q(y),q,\zeta)L(x,p,\zeta)\;,
\end{equation}
which according to \eqref{prf1} is equivalent to \eqref{refs0}. Furthermore, if $L$ is a strong Lax matrix of $R^{p,q}$, then equation \eqref{prf2} is equivalent to \eqref{prf1} which is  equivalent to 
\begin{equation*}
(\hat u,\hat v)=(\sigma_p^{-1} (u(x,p,\sigma_q(y),q)),v(x,p,\sigma_q(y),q))=\hat R^{\sigma}(x,y)\;.
\end{equation*}

The proof that the   YB map $\tilde R^{\sigma}$ satisfies (\ref{refs1}) is similar so we omit it.
\end{proof}

\begin{remark}
Equation \eqref{refs0} does not provide a Lax representation of the YB map  \eqref{Rs} in the classical sense of the definition associated with equation \eqref{laxmat},  as it involves two matrices,  $L:=L(x,\alpha,\zeta)$ and $M:=L(\sigma_{\alpha}(x),\alpha,\zeta)$, instead of one. However, we still somewhat loosely refer to $L$ and $M$ as Lax matrices, since each one is a Lax matrix of a YB map.
\end{remark}

\subsection{Rational Lax matrices for the $\mathcal{K},\Lambda, \mathcal{H}$ and $\mathcal{F}$ lists of non-abelian quadrirational YB maps}

Recently, there has been increased interest in non-abelian analogues of YB maps \cite{Adam1,Doli,KasKoul,Kass:2023,Rizos:2}. The  non-abelian extensions of the $F$ and the $H$ lists were obtained in  \cite{KasKoul}. In detail,   there were considered quadrirational YB maps defined on $\mathcal{A}\times \mathcal{A}$   where $\mathcal{A}$  a division ring, that served as the noncommutative analogues of the the $F$ and the $H$ lists. Moreover they were refered to as the $\mathcal{F}$ and the $\mathcal{H}$ list of non-abelian quadrirational YB maps. Furthermore, two additional lists of non-abelian quadrirational YB maps were obtained the so-called $\mathcal{K}$ and the $\Lambda$ lists (see Figure \ref{nonco}).

The following YB maps were obtained in \cite{KasKoul}
\begin{align}\label{K1}
\begin{aligned}
& \quad \mathcal{K}_{a,b,c}^{p,q}:(x,y)\mapsto (u,v),\qquad  \mbox{where}\\
 u=&y\left(a x y+bq-cq(x+y)\right)^{-1} \left(a x y+bp-c(qx+py)\right),\\
 v=&\left(a x y+bq-c(py+qx)\right)\left(axy+bp-cp(x+y)\right)^{-1}u,
 \end{aligned}
 \end{align}
 
\begin{align}\label{L1}
\begin{aligned}
& \quad \Lambda_{a,b,c}^{p,q}:=(\psi^{-1}\times id)\,\mathcal{K}_{a,b,c}\; (id \times \psi):(x,y)\mapsto (u,v),\qquad  \mbox{where}\\
 u=&py\left(ab(qx+py)-cq(bp+axy)\right)^{-1}\left(ab(x+y)-c(bq+axy)\right),\\
 v=&q\left(ab(x+y)-c(bp+axy)\right)\left(ab(qx+py)-cp(bq+axy)\right)^{-1}x,
 \end{aligned}
 \end{align}

\begin{align}\label{H1}
\begin{aligned}
& \quad \mathcal{H}_{a,b,c}^{p,q}:=(\phi^{-1}\times id)\,\mathcal{K}_{a,b,c}\; (id \times \phi):(x,y)\mapsto (u,v),\qquad  \mbox{where}\\
 u=&\left((axy-bq)(y-\frac{c}{a}q)^{-1}-(axy-bp)(y-\frac{b}{c})^{-1}\right)^{-1}\\
 &\qquad\left(p(axy-bq)(\frac{a}{c}y-q)^{-1}-(axy-bp)(\frac{c}{b}y-1)^{-1}\right)\\
 v=&\left(a (ab-c^2q)xy+abc(q-p)y+bq(c^2p-ab)\right)\\
  &\qquad\left(a (ab-c^2p)xy+abc(p-q)x+bp(c^2q-ab)\right)^{-1}x,
 \end{aligned}
 \end{align}

\begin{align}\label{F1}
\begin{aligned}
& \quad \mathcal{F}^{p,q}_{a,b,c}:=(\psi^{-1}\circ\phi^{-1}\times id)\,\mathcal{K}_{a,b,c}\; (id \times \phi\circ \psi):(x,y)\mapsto (u,v),\qquad  \mbox{where}\\
 u=&p\left(c p(x-y)(b-cy)^{-1}-a(qx-py)(cq-ay)^{-1}\right)^{-1}\\
  &\quad \qquad  \left(b(x-y)(b-cy)^{-1}-c(qx-py)(cq-ay)^{-1}\right),\\
 v=&q\left((ab-c^2q)x+bc(q-p)+(c^2p-ab)y\right)\\
 &\quad \qquad\left(q(ab-c^2p)x+ac(p-q)xy+p(c^2q-ab)y\right)^{-1}x,
 \end{aligned}
 \end{align}
 where $a,b$ and $c$ denote free parameters that take values in the center of the division ring  $\mathcal{A}$ and so does the Yang-Baxter parameters $p,q$.  
  Due to  equivalence relation (\ref{equiv}),  the  parameters $a,b,c$  can be scaled to $1$ when they  are neither $0$ nor $\infty$. So without loss of generality we can assume that   $a,b,c\in\{0,1,\infty\}.$
  The maps $\mathcal{K}_{a,b,c}, \Lambda_{a,b,c}, \mathcal{H}_{a,b,c},$ and  $\mathcal{F}_{a,b,c},$ for $a=b=c=1,$   are referred to as {\em the   generic maps (or members)} of the $\mathcal{K},$  $\Lambda$ $\mathcal{H},$ and $\mathcal{F}$ lists respectively.
 

\begin{figure}[htbp]
\begin{center}
 \begin{minipage}[htb]{0.35\textwidth}
\begin{tikzcd}[every arrow/.append style={}]
   &  \mathcal{K}\arrow{ld}[left,above,yshift=0.7ex]{\Phi} \arrow{dr}{\Psi}  & \\
\mathcal{H} \arrow{dr}[left,below,yshift=-0.3ex,xshift=-0.5ex]{\Psi} & & \Lambda\arrow{dl}{\Phi}\\
    & \mathcal{F}  &
\end{tikzcd}
\captionsetup{font=footnotesize}
\captionof*{figure}{(a) Non-abelian setting}
\end{minipage} 
\begin{minipage}[htb]{0.35\textwidth}
\begin{tikzcd}[every arrow/.append style={},row sep=12ex]
\mathcal{H}\simeq \mathcal{K} \simeq \Lambda \arrow{d}{\Phi\circ \Psi}\\
\mathcal{F}
\end{tikzcd}
\captionsetup{font=footnotesize}
\captionof*{figure}{(b) Abelian setting}
\end{minipage}
\caption{The $\mathcal{F},$  $\mathcal{H},$ $\mathcal{K}$ and  $\Lambda$ lists of quadrirational Yang-Baxter maps in the  non-abelian and in the abelian setting. The generic members of these lists are related by the morphisms  $\Phi: R\rightarrow (\phi^{-1}\times id)R(id\times \phi)$ and $\Psi: R\rightarrow (\psi^{-1}\times id )R(id\times \psi),$ where $\phi,\psi,$  symmetries.} \label{nonco}
\end{center}
\end{figure}
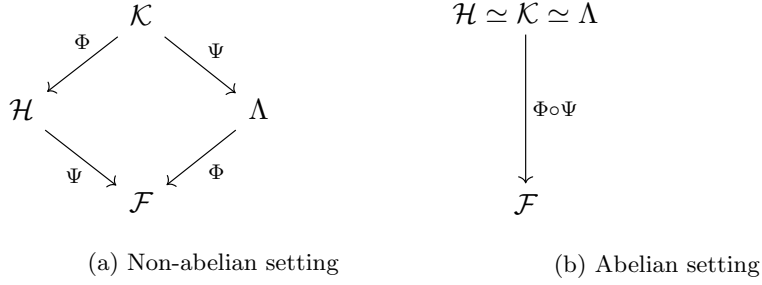

 These generic maps admit the dihedral  discrete symmetry group $D_2=<\alpha,\beta | \alpha^2=\beta^2=(\alpha\beta)^2=Id>,$ that is, the following families of bijections
\begin{align} \label{symab}
    \phi(p):x\mapsto& \frac{b}{a}(a x-c p)(cx-b)^{-1},&  \psi(p):x\mapsto& \frac{b}{a}p x^{-1},
\end{align}
serve as a birational realization of $D_2.$ 

 A strong Lax matrix was associated to the generic member of  the  $\mathcal{K}$ list in \cite{KasKoul}.  Theorem \ref{propsym}  allows to obtain Lax matrices for  the generic members of the   $\Lambda,$ $\mathcal{H}$ and $\mathcal{F}$ lists, out of the strong Lax matrix of the generic map of the $\mathcal{K}$ list, as the following Proposition suggests. 
\begin{proposition}
Assuming that the spectral parameter $\zeta$ and the Yang-Baxter parameters $p,q$ are elements of the center of the division ring $\mathcal{A},$ the following statements hold.
\begin{enumerate}
    \item The generic map of the $\mathcal{K}$ list, that is mapping  (\ref{K1}), is equivalent to the refactorization problem
    \begin{align*} 
L(u,p,\zeta)L(v,q,\zeta)=L(y,q,\zeta)L(x,p,\zeta),
\end{align*}
where
\begin{align}\label{laxk1}
 L(x,p,\zeta):=\begin{pmatrix}
ax-c p&\zeta(b-cx)\\
a-cp x^{-1}&p(bx^{-1}-c)   
\end{pmatrix};
\end{align}

\item The generic map of the $\Lambda$ list, that is mapping  (\ref{L1}), is equivalent to the refactorization problem
    \begin{align*} 
M(u,p,\zeta)L(v,q,\zeta)=M(y,q,\zeta)L(x,p,\zeta),
\end{align*}
where $L$ is given in (\ref{laxk1}) and
\begin{align*}
 M(x,p,\zeta):=\begin{pmatrix}
p(bx^{-1}-c)&\zeta b(1-\frac{c}{a}px^{-1})\\
a(1-\frac{c}{b}x)&ax-pc   
\end{pmatrix};
\end{align*}

\item The generic map of the $\mathcal{H}$ list, that is mapping  (\ref{H1}), is equivalent to the refactorization problem
    \begin{align*} 
M(u,p,\zeta)L(v,q,\zeta)=M(y,q,\zeta)L(x,p,\zeta),
\end{align*}
where $L$ is given in (\ref{laxk1}) and
\begin{align*}
 {\displaystyle M(x,p,\zeta):=\begin{pmatrix}
b(ax-cp)(cx-b)^{-1}-cp&\zeta b\left(1-\frac{c}{a}(ax-cp)(cx-b)^{-1}\right)\\
a\left(1-p\frac{c}{b}(cx-b)(ax-cp)^{-1}\right)&p\left(a(cx-b)(ax-cp)^{-1}-c\right)   
\end{pmatrix}};
\end{align*}

\item The generic map of the $\mathcal{F}$ list, that is mapping  (\ref{F1}), is equivalent to the refactorization problem
    \begin{align*} 
M(u,p,\zeta)L(v,q,\zeta)=M(y,q,\zeta)L(x,p,\zeta),
\end{align*}
where $L$ is given in (\ref{laxk1}) and
\begin{align*}
 {\displaystyle M(x,p,\zeta):=\begin{pmatrix}
p \left(a(cx-b)(ax-cp)^{-1}-c\right)&\zeta\left(b-cp(cx-b)(ax-cp)^{-1}\right)\\
a-c(ax-cp)(cx-b)^{-1}&b(ax-cp)(cx-b)^{-1}-cp   
\end{pmatrix}}.
\end{align*}

\end{enumerate}
\end{proposition}

\begin{proof}
The first item of the Proposition was proven in \cite{KasKoul}, where it was shown that $\mathcal{K}$ admits the strong Lax matrix (\ref{laxk1}).  Then, using the fact that the generic maps of the $\mathcal{K},$ $\Lambda,$ $\mathcal{H}$ and $\mathcal{F}$ lists are related via the symmetries (\ref{symab}), the proof of the remaining items of the Proposition are a direct consequence of Theorem \ref{propsym}.
\end{proof}

\begin{corollary}
By specifying the constants $a,b,c\in\{0,1,\infty\}$ appropriately  in (\ref{K1})-(\ref{F1}), we obtain  members of the associated  $\mathcal{K},$  $\Lambda$, $\mathcal{H}$ and $\mathcal{F}$ lists  of Yang-Baxter maps. For example,   the non-abelian extensions of $F_I, F_{II}$ and $F_{III}$ respectively are $\mathcal{F}_I:=\mathcal{F}^{p,q}_{1,1,1},$ $\mathcal{F}_{II}:=\mathcal{F}^{p,q}_{0,1,1}$ and $\mathcal{F}_{III}:=\mathcal{F}^{p,q}_{0,0,1}.$ Also, $\mathcal{F}_{IV}$ is obtained from $\mathcal{F}_{II}$ and $\mathcal{F}_{V}$ is obtained from $\mathcal{F}_{IV}$ by limiting procedures (see \cite{KasKoul}).  

Assuming that all variables that participate to the aforementioned maps commute, $\mathcal{K},$  $\Lambda$ $\mathcal{H},$ and $\mathcal{F}$ lists collapse to the $H$ and $F$ lists of YB maps (see  Figure \ref{nonco}). That allows us to obtain the rational Lax matrices associated with the $F$-list of quadrirational Yang-Baxter maps which are presented in Table \ref{table2}. Note that   we  were  not able to obtain rational Lax matrices associated with $F_{IV}$ since the latter does not admit any    symmetry. Mapping $F_{IV}$ admits the Lax matrix (\ref{laxf4}) that is not rational so we did not include it in Table.  \ref{table2}.
\end{corollary}
  
\begin{table}[htbp]
 \begin{tabular}{|c|c|c|c|c|}
  \hline
   & $F_I$&$F_{II}$&$F_{III}$&$F_V$ \\  \hline
  $L(x,\alpha;\zeta)$ & ${\ds\begin{pmatrix}
                          \frac{x-\alpha}{x-1}&\zeta (\alpha-1)\frac{x}{x-1}\\
                          \alpha-x& x
                          \end{pmatrix}}$ & ${\ds\begin{pmatrix}
                          1&\zeta \alpha (x-1) \\
                          \frac{1}{x}& \frac{x-1}{x}\end{pmatrix}}$& ${\ds\begin{pmatrix}
                          1&\zeta \alpha x \\
                          \frac{1}{x}& 1 \end{pmatrix}}$  &  ${\ds\begin{pmatrix}
                          x&\alpha-\zeta+x^2 \\
                          1 & x \end{pmatrix}}$    \\  \hline
  $M(x,\alpha;\zeta)$ & ${\ds\begin{pmatrix}
                          \alpha\frac{x-1}{x-\alpha}&-\zeta \alpha (\alpha-1)\frac{1}{x-\alpha}\\
                          \alpha \frac{x-1}{x}& \frac{\alpha}{x}
                          \end{pmatrix}}$ & ${\ds\begin{pmatrix}
                          1&-\zeta \alpha x \\
                          \frac{1}{1-x}& \frac{x}{x-1}\end{pmatrix}}$& ${\ds\begin{pmatrix}
                          1&-\zeta \alpha x \\
                          -\frac{1}{x}& 1 \end{pmatrix}}$  &  ${\ds\begin{pmatrix}
                          -x&\alpha-\zeta+x^2 \\
                          1 & -x \end{pmatrix}}$    \\

  \hline
\end{tabular}
\caption{Lax matrices of the $F-$list of quadrirational YB maps} \label{table2}
\end{table}


\section{Discussion}

In this article we have provided rational Lax matrices for the generic members of the $\mathcal{F},$  $\mathcal{H},$ $\mathcal{K}$ and  $\Lambda$ lists of non-abelian quadrirational Yang-Baxter maps. In the abelian setting,  we obtained rational Lax matrices for the abelian quadrirational YB maps of  the  $H$ and the $F$ lists.

As already mentioned, due to the lack of classification results up to the equivalence relation introduced in Proposition \ref{equivalence}, the $H$ and the $F$ lists do not exhaust all quadrirational abelian YB maps and the same holds true for the  $\mathcal{F},$  $\mathcal{H},$ $\mathcal{K}$ and  $\Lambda$ lists. A possible direction for future research is to complete the classification of quadrirational YB maps at least in the abelian setting.  For example there exist YB maps which are quadrirational but not equivalent with any member of the $F$ or the $H$ list. One of such maps together with its associated Lax matrix was firstly introduced in \cite{kp4}, as a four-parameter extension of the $H_{III}^A$ YB map. This map turned out  not to be an involution, which explains why it is excluded from the $H$ and $F$ lists (the involutivity property of a map is preserved under Yang–Baxter equivalence). 

In what follows, we present four  non-abelian  six-parameter extensions of $\mathcal{F}_{III},$ $\mathcal{K}_{III},$ $\mathcal{F}_{V}$ and $\mathcal{H}_{V}$ maps. We denote these maps as  ${}^e\mathcal{F}_{III},$ ${}^e\mathcal{K}_{III},$ ${}^e\mathcal{F}_{V}$ and ${}^e\mathcal{H}_{V}$  and since they turn out to be non involutive, they are not equivalent to their two-parameter counterparts.
Explicitly they read $R^{{\bf p},{\bf q}}:(x,y)\mapsto(u,v),$ where  ${\bf p}:=(p, p^{(1)},p^{(2)}),$ ${\bf q}:=(q, q^{(1)},q^{(2)})$ and
\begin{align*}
\begin{aligned}
u=&p^{-1}y\left(p^{(2)}x-q^{(1)}y\right)^{-1}\left(p q^{(2)} x-qp^{(1)}y\right),\\
v=&q^{-1}x\left(p^{(2)}x-q^{(1)}y\right)^{-1}\left(p q^{(2)} x-qp^{(1)}y\right),
\end{aligned}
&&  ({}^e\mathcal{F}_{III})
\end{align*}
\begin{align*}
\begin{aligned}
u=&q^{(1)}-p^{(1)}+y+(p-q+q^{(1)}q^{(2)}-p^{(1)}p^{(2)})\left(p^{(1)}-q^{(2)}+x-y\right)^{-1},\\
v=&p^{(2)}-q^{(2)}+x+(p-q+q^{(1)}q^{(2)}-p^{(1)}p^{(2)})\left(p^{(1)}-q^{(2)}+x-y\right)^{-1},
\end{aligned}
&& ({}^e\mathcal{F}_{V})
\end{align*}
and

\begin{align*}
\begin{aligned}
u=&p^{-1}y\left(p^{(2)}x+q^{(1)}y\right)^{-1}\left(p q^{(2)} x+qp^{(1)}y\right),\\
v=&q^{-1}x\left(p^{(2)}x+q^{(1)}y\right)^{-1}\left(p q^{(2)} x+qp^{(1)}y\right),
\end{aligned}
&&  ({}^e\mathcal{K}_{III})
\end{align*}
\begin{align*}
\begin{aligned}
u=&q^{(1)}+p^{(1)}+y-(p-q+q^{(1)}q^{(2)}-p^{(1)}p^{(2)})\left(p^{(1)}+q^{(2)}+x+y\right)^{-1},\\
v=&p^{(2)}+q^{(2)}+x+(p-q+q^{(1)}q^{(2)}-p^{(1)}p^{(2)})\left(p^{(1)}+q^{(2)}+x+y\right)^{-1},
 \end{aligned}& &  ({}^e\mathcal{H}_{V})
\end{align*}
where for ${}^e\mathcal{H}_{V}$ the YB parameters transform as
\begin{align*}(p,p^{(1)},p^{(2)};q,q^{(1)},q^{(2)})
\mapsto (p,-p^{(1)},-p^{(2)};q,-q^{(1)},-q^{(2)}).
\end{align*}
The non-abelian YB map ${}^e\mathcal{K}_{III}$ was firstly introduced in \cite{Kass:2021}. In its abelian limit and for $p=q=1,$ coincides with the two-parameter extension of $H_{III}^A$ given in \cite{kp4}. Note that in \cite{Kass:2023} the non-Abelian hierarchies of the $\mathcal{K}$ and the $\Lambda$ lists were obtained. 

We will now present explicitly the Lax matrices of the aforementioned non-abelian maps of this Section.  The YB map ${}^e\mathcal{K}_{III}$ is equivalent to the refactorization problem
   $ 
L(u,{\bf p},\zeta)L(v,{\bf q},\zeta)=L(y,{\bf q},\zeta)L(x,{\bf p},\zeta),
$
where
$
 L(x,{\bf p},\zeta):=\begin{pmatrix}
p^{(1)}&\zeta x\\
p x^{-1}&p^{(2)}   
\end{pmatrix},
$
while ${}^e\mathcal{F}_{III}$ is equivalent to the refactorization problem
    $ 
M(u,{\bf p},\zeta)L(v,{\bf q},\zeta)=M(y,{\bf q},\zeta)L(x,{\bf p},\zeta),
$
where
$
 L(x,{\bf p},\zeta):=\begin{pmatrix}
p^{(1)}&\zeta x\\
p x^{-1}&p^{(2)}   
\end{pmatrix},$ $M(x,{\bf p},\zeta):=\begin{pmatrix}
p^{(1)}&-\zeta x\\
-p x^{-1}&p^{(2)}   
\end{pmatrix}.
$
Finally, ${}^e\mathcal{H}_{V}$ is equivalent to the refactorization problem
    \begin{align*} 
L(u,p,-p^{(1)},-p^{(2)},\zeta)L(v,q,-q^{(1)},-q^{(2)},\zeta)=L(y,q,q^{(1)},q^{(2)},\zeta)L(x,p,p^{(1)},p^{(2)},\zeta),
\end{align*}
where
$
 L(x,{\bf p},\zeta):=\begin{pmatrix}
x+p^{(1)}&x^2+(p^{(1)}+p^{(2)})x+p-\zeta \\
1 &x+p^{(2)}   
\end{pmatrix},
$
while ${}^e\mathcal{F}_{V}$ is equivalent to the refactorization problem
 $ 
M(u,{\bf p},\zeta)L(v,{\bf q},\zeta)=M(y,{\bf q},\zeta)L(x,{\bf p},\zeta),
$
where
\begin{align*}
 L(x,{\bf p},\zeta):=\begin{pmatrix}
-x-p^{(1)}&x^2+(p^{(1)}+p^{(2)})x+p-\zeta \\
1 &-x-p^{(2)}   
\end{pmatrix},\\
M(x,{\bf p},\zeta):=\begin{pmatrix}
x+p^{(1)}&x^2+(p^{(1)}+p^{(2)})x+p-\zeta \\
1 &x+p^{(2)}   
\end{pmatrix}.
\end{align*}

We leave the study of the non-abelian multi-parametric extensions of YB maps to a future study.

\section*{Acknowledgements}
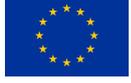
\parbox{.135\textwidth}{\begin{tikzpicture}[scale=.03]
\fill[fill={rgb,255:red,0;green,51;blue,153}] (-27,-18) rectangle (27,18);
\pgfmathsetmacro\inr{tan(36)/cos(18)}
\foreach \i in {0,1,...,11} {
\begin{scope}[shift={(30*\i:12)}]
\fill[fill={rgb,255:red,255;green,204;blue,0}] (90:2)
\foreach \x in {0,1,...,4} { -- (90+72*\x:2) -- (126+72*\x:\inr) };
\end{scope}}
\end{tikzpicture}} \parbox{.85\textwidth}
{This research is part of the project No. 2022/45/P/ST1/03998  co-funded by the National Science Centre and the European Union Framework Programme
 for Research and Innovation Horizon 2020 under the Marie Sklodowska-Curie grant agreement No. 945339. For the purpose of Open Access, the author has applied a CC-BY public copyright licence to any Author Accepted Manuscript (AAM) version arising from this submission.}

\appendix
\appendixpage
\section{The $F$ and the $H-$list of quadrirational Yang-Baxter maps} \label{app1}

The Yang-Baxter maps $R$ of the $F$ and the $H-$list, explicitly read:
$$
R: \mathbb{CP}^1\times \mathbb{CP}^1\ni (x,y)\mapsto (u,v)\in \mathbb{CP}^1\times \mathbb{CP}^1
$$
where:
\begin{align*}
&\begin{aligned}
u=&p y P,\\
v=&q x P,
\end{aligned}
&&P=\frac{(1-q)x+q-p+(p-1)y}{q(1-p)x+(p-q)xy+p (q-1) y},&&(F_I)\\
&\begin{aligned}
u=&\frac{y}{p} P,\\
v=&\frac{x}{q} P,
\end{aligned}
&&P=\frac{p x-q y+q-p}{x-y},&&(F_{II}),  \\
&\begin{aligned}
u=&\frac{y}{p} P,\\
v=&\frac{x}{q} P,
\end{aligned}
&&P=\frac{p x-q y}{x-y},&&(F_{III}),\\
&\begin{aligned}
u=&y P,\\
v=&x P,
\end{aligned}
&&P=1+\frac{q-p}{x-y},&&(F_{IV}), \\
&\begin{aligned}
u=&y+ P,\\
v=&x+ P,
\end{aligned}
&&P=\frac{p-q}{x-y},&&(F_{V}), 
\end{align*}

and 
\begin{align*}
&\begin{aligned}
u=&yQ,\\
v=&xQ^{-1},
\end{aligned} &
&Q=\frac{(p-1)xy+(q-p)x+p(1-q)}{(q-1)xy+(p-q)y+q(1-p)},& &(H_I)\\ 
&\begin{aligned}
u=&\frac{q}{p}y+\frac{1}{p}Q,\\
v=&\frac{p}{q}x-\frac{1}{q}Q,
\end{aligned} &
&Q=\frac{(p-q)xy}{x+y-1},& &(H_{II}) \\
&\begin{aligned}
u=&\frac{y}{p} Q,\\
v=&\frac{x}{q} Q,
\end{aligned} &
&Q=\frac{p x+q y}{x+y},& &(H_{III}^A)\\
&\begin{aligned}
u=&y Q,\\
v=&x Q^{-1},
\end{aligned} &
&Q=\frac{1+q x y}{1+p x y},& &(H_{III}^B)\\
&\begin{aligned}
u=&y-Q,\\
v=&x+Q,
\end{aligned} &
&Q=\frac{p-q}{x+y}.& &(H_{V})\\
\end{align*}

The symmetries of the $F$ and the $H-$list
are listed below.
\begin{align*}
    \phi(p):x\mapsto& \frac{x-p}{x-1},& \psi(p):x\mapsto& \frac{p}{x},& (F_I),(H_I)\\
    \phi(p):x\mapsto&1-x, & & & (F_{II}),(H_{II})\\
     \phi(p):x\mapsto& \frac{1}{px},& \psi(p):x\mapsto& -x,& (F_{III}),(H_{III}^A),(H_{III}^B)\\
     \phi(p):x\mapsto&-x, & & & (F_{V}),(H_{V})
\end{align*}
The Yang-Baxter map $F_{IV}$ does not admit a symmetry.

\end{document}